\documentclass[pra,twocolumn,amsmath,18pt]{revtex4-1}
\usepackage{pifont}
\usepackage{graphicx,epsfig,subfigure,dsfont,amssymb,amsmath,amsthm,amsfonts,amsbsy,mathrsfs,amscd}
\usepackage[all]{xy}

\newtheorem{theorem}{Theorem}

\newtheorem{lemma}{Lemma}

\input amssym.def

\begin{document}

\title{Tighter Monogamy Inequalities of Multiqubit Entanglement}

\smallskip
\author{Jia-Yi Li$^1$ }
\email{17375870834@163.com}
\author{Zhong-Xi Shen$^1$}
\email{18738951378@163.com}
\author{Shao-Ming Fei$^1$}
\email{feishm@cnu.edu.cn}
\affiliation{
$^1$School of Mathematical Sciences, Capital Normal University, Beijing 100048, China\\
}

\begin{abstract}
Multipartite entanglement holds great importance in quantum information processing. The distribution of entanglement among subsystems can be characterized by monogamy relations. Based on the $\beta$th power of concurrence and negativity, we provide two new monogamy inequalities. 
Through detailed examples, we demonstrate that these inequalities are tighter than previous results.
\end{abstract}

\maketitle
\baselineskip20pt
\section{Introduction}

As a fundamental resource \cite{RPMK,NMR} in quantum communications \cite{JB,RH,NR}, quantum cryptography \cite{AKE,NGWH} and quantum computing \cite{AR,MAIL,ASC}, quantum entanglement \cite{RPMK,IK,OG,STMA} has attracted much attention for a long time. 
One notable characteristic of quantum entanglement is that when a quantum system becomes entangled with one of its subsystems, it restricts its entanglement with the rest of the subsystems. This phenomenon is referred to as the monogamy of entanglement (MoE) \cite{BMT,JSK}. The concept of MoE is crucial in numerous quantum information and communication processing endeavors.

For a tripartite quantum state $\rho_{ABC}$, MoE can be described by 
\begin{equation}
E(\rho_{A|BC})\geq E(\rho_{AB})+E(\rho_{AC}),
\end{equation}
where $\rho_{AB}$ and $\rho_{AC}$ are the reduced density matrices, $E$ is a  entanglement measure, $E(\rho_{A|BC})$ stands for the entanglement under the bipartition $A$ and $BC$. Such monogamy inequalities depend on the entanglement measure $E$ and the state $\rho_{ABC}$. 

Recently, generalized classes of monogamy inequalities have been proposed, which refer to the $\beta$-th power of the entanglement measures. In Ref.\cite{SM.Fei1,SM.Fei3}, the authors showed that the squared concurrence and CREN satisfy the monogamy inequalities in multiqubit systems for $\beta\geq2$. The authors in Ref.\cite{Yanglongmei,Taoyuanhong} presented the tighter monogamy inequalities for $\beta\geq2$ and $\beta\geq4$.

In this work, we study monogamy relations in multiqubit systems. We establish new monogamy inequalities for multiqubit entanglement, which are tighter than existing monogamy inequalities. Thus, our monogamy inequalities lead to a finer characterization of the distribution of entanglement in multiqubit systems.

\section{Tighter monogamy inequalities related to concurrence}
The concurrence of a bipartite pure state $|\psi\rangle_{AB}$ in Hilbert space $H_A\otimes H_B$ is defined by \cite{fide,SM.Fei4}
$\mathcal{C}(|\psi\rangle_{AB})=\sqrt{2(1-{\rm tr}\rho_A^2)}$ with $\rho_A={\rm tr}_B|\psi\rangle_{AB}\langle\psi|$.
 The concurrence of a bipartite mixed state $\rho_{AB}$ is defined by the convex roof extension,
$\mathcal{C}(\rho_{AB})=\min\limits_{\{p_i,|\psi_i\rangle\}}\sum\limits_{i}p_i\mathcal{C}(|\psi_i\rangle)$,
where the minimum is taken over all possible pure state decompositions of $\rho_{AB}=\sum\limits_{i}p_i|\psi_i\rangle\langle\psi_i|$
with $\sum p_i=1$ and $p_i\geq0$. For an $N$-qubit state $\rho_{AB_1\cdots B_{N-1}}\in H_A\otimes H_{B_1}\otimes\cdots\otimes H_{B_{N-1}}$,
the concurrence $\mathcal{C}(\rho_{A|B_1\cdots B_{N-1}})$ of the state $\rho_{AB_1\cdots B_{N-1}}$ under bipartite partition $A$ and $B_1\cdots B_{N-1}$ satisfies \cite{SM.Fei1}
\begin{equation}\label{Con1}
\begin{array}{rl}
&\mathcal{C}^\beta(\rho_{A|B_1\cdots B_{N-1}})\\[2.0mm]
&\ \ \geq\mathcal{C}^\beta(\rho_{AB_1})+\mathcal{C}^\beta(\rho_{AB_2})+\cdots+\mathcal{C}^\beta(\rho_{AB_{N-1}}),
\end{array}
\end{equation}
for $\beta\geq2$, where $\rho_{AB_j}$ denote the two-qubit reduced density matrices of subsystems $AB_j$ for $j=1,2,\ldots,N-1$.
Later, the relation \eqref{Con1} is improved to be \cite{SM.Fei2},
\begin{equation}\label{Con2}
\begin{array}{rl}
&\mathcal{C}^\beta(\rho_{A|B_1\cdots B_{N-1}})\\[2.0mm]
&\ \ \geq\mathcal{C}^\beta(\rho_{AB_1})+\frac{\beta}{2}\mathcal{C}^\beta(\rho_{AB_2})+\cdots\\[2.0mm]
&\ \ \ \ +\big(\frac{\beta}{2}\big)^{m-1}\mathcal{C}^\beta(\rho_{AB_{m}})\\[2.0mm]
&\ \ \ \ +\big(\frac{\beta}{2}\big)^{m+1}[\mathcal{C}^\beta(\rho_{AB_{m+1}})+\cdots+\mathcal{C}^\beta(\rho_{AB_{N-2}})]\\[2.0mm]
&\ \ \ \ +\big(\frac{\beta}{2}\big)^m\mathcal{C}^\beta(\rho_{AB_{N-1}})
\end{array}
\end{equation}
for $\beta\geq2$, if $\mathcal{C}(\rho_{AB_i})\geq\mathcal{C}(\rho_{A|B_{i+1}\cdots B_{N-1}})$ for $i=1,2,\ldots,m$,
and $\mathcal{C}(\rho_{AB_j})\leq\mathcal{C}(\rho_{A|B_{j+1}\cdots B_{N-1}})$ for $j=m+1,\ldots,N-2$.

The relation \eqref{Con2} is further improved for $\beta\geq2$ as \cite{SM.Fei3},
\begin{equation}\label{Con3}
\begin{array}{rl}
&\mathcal{C}^\beta(\rho_{A|B_1\cdots B_{N-1}})\\[2.0mm]
&\ \ \geq\mathcal{C}^\beta(\rho_{AB_1})+\big(2^{\frac{\beta}{2}}-1\big)\mathcal{C}^\beta(\rho_{AB_2})+\cdots\\[2.0mm]
&\ \ \ \ +\big(2^{\frac{\beta}{2}}-1\big)^{m-1}\mathcal{C}^\beta(\rho_{AB_{m}})\\[2.0mm]
&\ \ \ \ +\big(2^{\frac{\beta}{2}}-1\big)^{m+1}[\mathcal{C}^\beta(\rho_{AB_{m+1}})+\cdots+\mathcal{C}^\beta(\rho_{AB_{N-2}})]\\[2.0mm]
&\ \ \ \ +\big(2^{\frac{\beta}{2}}-1\big)^m\mathcal{C}^\beta(\rho_{AB_{N-1}})
\end{array}
\end{equation}
under the same conditions as in \eqref{Con2}. The relation \eqref{Con3} is further improved to be
\begin{equation}\label{ylm}
\begin{array}{rl}
&\mathcal{C}^\beta_{A|B_1\cdots B_{N-1}}\\[2.0mm]
&\ \ \geq\mathcal{C}^\beta_{AB_1}+\frac{(1+k)^{\frac{\beta}{2}}-1}{k^{\frac{\beta}{2}}}\mathcal{C}^\beta_{AB_2}+\cdots\\[2.0mm]
&\ \ \ \ +\Big(\frac{(1+k)^{\frac{\beta}{2}}-1}{k^{\frac{\beta}{2}}}\Big)^{m-1}\mathcal{C}^\beta_{AB_{m}}\\[2.0mm]
&\ \ \ \ +\Big(\frac{(1+k)^{\frac{\beta}{2}}-1}{k^{\frac{\beta}{2}}}\Big)^{m+1}\Big(\mathcal{C}^\beta_{AB_{m+1}}+\cdots+
\mathcal{C}^\beta_{AB_{N-2}}\Big)\\[2.0mm]
&\ \ \ \ +\Big(\frac{(1+k)^{\frac{\beta}{2}}-1}{k^{\frac{\beta}{2}}}\Big)^m\mathcal{C}^\beta_{AB_{N-1}}
\end{array}
\end{equation}
for all $\beta\geq2$ with $k\mathcal{C}_{AB_i}^2\geq\mathcal{C}_{A|B_{i+1}\cdots B_{N-1}}^2$
for $i=1,2,\ldots,m$, and $\mathcal{C}_{AB_j}^2\leq k\mathcal{C}_{A|B_{j+1}\cdots B_{N-1}}^2$ for $j=m+1,\ldots,N-2$,
$\forall 1\leq m\leq N-3$, $N\geq 4$.

In fact, the monogamy relation \eqref{Con3} is based on the inequality \cite{SM.Fei3},
$(1+t)^x \geq 1+(2^x-1)t^x$, for any real numbers $x$ and $t$ satisfying $0\leq t \leq 1$ and $x\geq 1$.
Recently, in \cite{Taoyuanhong} the authors proved that
$(1+t)^x \geq 1+(2^x-t^x)t^x$ and from which derived a new monogamy relation,
\begin{equation}\label{Con4}
\begin{aligned}
\mathcal{C} ^{\beta}_{A|B_0\cdots B_{N-1}} & \geq \sum\limits_{i=0}^m ( \prod\limits_{j=0}^i \mathcal{M}_j ) \mathcal{C}^{\beta}_{AB_i}   \\
 & +( \prod\limits_{i=1}^{m+1} \mathcal{M}_i ) ( \prod\limits_{j=m+1}^{N-2} \mathcal{Q}_j \mathcal{C}^\beta_{AB_j} + \mathcal{C}^{\beta}_{AB_{N-1}} )
\end{aligned}
\end{equation}
for $\beta \geq 4$, where $\mathcal{M}_0=1$, $\mathcal{M}_{i+1}=2^{\frac{\alpha}{2}} - \frac{\mathcal{C} ^\beta_{A|B_{i+1}\cdots B_{N-1}}}{\mathcal{C}^\beta_{AB_i}} $ for $i=0,1,2,\cdots m$, and $\mathcal{Q}_j=2^{\frac{\alpha}{2}}-\frac{\mathcal{C}^\beta_{AB_j}}{\mathcal{C} ^\beta_{A|B_{j+1}\cdots B_{N-1}}} $ for $j=m+1, \cdots , N-2$. The monogamy inequality (\ref{Con4}) holds under the conditions that $\mathcal{C}^{\beta}_{AB_i} \geq \mathcal{C} ^\beta_{A|B_{i+1}\cdots B_{N-1}}$ for $0 \leq i \leq m$, and $\mathcal{C}^{\beta}_{AB_j} \leq \mathcal{C} ^\beta_{A|B_{j+1}\cdots B_{N-1}}$ for $m+1 \leq j \leq N-3$, $0 \leq m \leq N-3$ and $N \geq 3$. In the relation \eqref{Con4}, $\mathcal{C}_{AB}=\mathcal{C}(\rho_{AB_i})$ denotes the concurrence of $\rho_{AB_i}=Tr_{B_i \cdots B_{i-1}B_{i+1} \cdots B_{N-1}} (|\psi\rangle_{AB_1 \cdots B_{N-1}} \langle\psi|)$, $\mathcal{C}_{A|B_1\cdots B_{N-1}}=\mathcal{C} (\rho_{A|B_1\cdots B_{N-1}})$.

The above monogamy relations for concurrence can be further tightened under certain conditions. For this purpose, let us first introduce the following lemma.
\begin{lemma}\label{con1}
Let $k$ be a real number satisfying $0<k\leq1$. For any $0\leq t\leq k$ and non-negative real number $x$, we have
\begin{equation}\label{Con5}
(1+t)^x\geq1+[\frac{(1+k)^x-1}{k^x}+k^x-t^x]t^x
\end{equation}
for $x\geq2$.
\end{lemma}

\begin{proof}
 Consider the function $f(x,y)=(1+y)^x-y^x+y^{-x}$ with $x\geq 1$ and $y\geq1$. We have $\frac{\partial f}{\partial y}=x(1+y)^{x-1} - xy^{x-1} -xy^{-x-1}$ and $\frac{\partial^2 f}{\partial y^2}=x(x-1)(1+y)^{x-2} - x(x-1)y^{x-2} +x(x+1)y^{-x-2}$. Therefore, $\frac{\partial^2 f}{\partial y^2} \geq 0$ for $x \geq 2$ and
$f(x,y)$ is an increasing function of $y$. Setting $y=\frac{1}{t}$, for $0\leq t\leq k$ we have $f(x, \frac{1}{t}) \geq f(x,\frac{1}{k})$, which gives rise to (\ref{Con5}). For $t=0$ the inequality (\ref{Con5}) is just trivial.
\end{proof}

\begin{lemma}\label{con2}
Let $k$ be a real number satisfying $0< k\leq1$. For any $2\otimes2\otimes2^{n-2}$ mixed state $\rho\in H_A\otimes H_B\otimes H_C$,
if $\mathcal{C}_{AC}^2\leq k\mathcal{C}_{AB}^2$, we have
\begin{equation}\label{Con9}
\mathcal{C}_{A|BC}^\beta\geq\mathcal{C}_{AB}^\beta+\Big[\frac{(1+k)^{\frac{\beta}{2}}-1}{k^{\frac{\beta}{2}}} + k^{\frac{\beta}{2}} -(\frac{\mathcal{C}_{AC}}{\mathcal{C}_{AB}})^\beta \Big]\mathcal{C}_{AC}^\beta
\end{equation}
for all $\beta\geq 4$.
\end{lemma}

\begin{proof}
Since $\mathcal{C}_{AC}^2\leq k\mathcal{C}_{AB}^2$ and $\mathcal{C}_{AB}>0$,
we obtain
\begin{equation}
\begin{array}{rl}
&\mathcal{C}_{A|BC}^\beta\geq(\mathcal{C}_{AB}^2+\mathcal{C}_{AC}^2)^{\frac{\beta}{2}}\\[1.5mm]
&\ \ \ \ \ \ \ \ \ =\mathcal{C}_{AB}^\beta \Big(1+\frac{\mathcal{C}_{AC}^2}{\mathcal{C}_{AB}^2}\Big)^{\frac{\beta}{2}}\\[1.5mm]
&\ \ \ \ \ \ \ \ \ \geq\mathcal{C}_{AB}^\beta \Big\{1+\Big[\frac{(1+k)^{\frac{\beta}{2}}-1}{k^{\frac{\beta}{2}}}+k^{\frac{\beta}{2}} -
\Big(\frac{\mathcal{C}_{AC}^2}{\mathcal{C}_{AB}^2}\Big)^{\frac{\beta}{2}}\Big]  \Big(\frac{\mathcal{C}_{AC}^2}{\mathcal{C}_{AB}^2}\Big)^{\frac{\beta}{2}} \Big\}\\[3mm]
&\ \ \ \ \ \ \ \ \ =\mathcal{C}_{AB}^\beta+\Big[\frac{(1+k)^{\frac{\beta}{2}}-1}{k^{\frac{\beta}{2}}} + k^{\frac{\beta}{2}} -(\frac{\mathcal{C}_{AC}}{\mathcal{C}_{AB}})^\beta \Big]\mathcal{C}_{AC}^\beta,
\end{array}
\end{equation}
where the first inequality is due to  $\mathcal{C}_{A|BC}^2\geq\mathcal{C}_{AB}^2+\mathcal{C}_{AC}^2$ for
arbitrary $2\otimes2\otimes2^{n-2}$ tripartite state $\rho_{ABC}$ \cite{XJR}, the second inequality is due to Lemma \ref{con1}.
It is seen that if $\mathcal{C}_{AB}=0$,  $\mathcal{C}_{AC}=0$ and the lower bound becomes trivially zero.
\end{proof}

By using the above lemma, for multiqubit systems we have the following Theorems.

\begin{theorem}\label{concurrence1}
Let $k$ be a real number satisfying $0< k\leq1$. For an $N$-qubit mixed state $\rho_{AB_1\cdots B_{N-1}}$, if $k\mathcal{C}_{AB_i}^2\geq\mathcal{C}_{A|B_{i+1}\cdots B_{N-1}}^2$
for $i=1,2,\ldots,m$, and $\mathcal{C}_{AB_j}^2\leq k\mathcal{C}_{A|B_{j+1}\cdots B_{N-1}}^2$ for $j=m+1,\ldots,N-2$,
$\forall 1\leq m\leq N-3$, $N\geq 4$, then we have
\begin{equation}\label{Con12}
\begin{aligned}
\mathcal{C} ^{\beta}_{A|B_1\cdots B_{N-1}} & \geq \sum\limits_{i=1}^{m-1} \Big( \prod\limits_{j=1}^i \mathcal{M}_j \Big) \mathcal{C}^{\beta}_{AB_{i+1}}   \\
& +\Big( \prod\limits_{i=1}^{m} \mathcal{M}_i \Big) \Big( \prod\limits_{j=m+1}^{N-2} \mathcal{Q}_j \mathcal{C}^\beta_{AB_j} + \mathcal{C}^{\beta}_{AB_{N-1}} \Big)
\end{aligned}
\end{equation}
for all $\beta\geq 4$, where $M=\frac{(1+k)^{\frac{\beta}{2}}-1}{k^{\frac{\beta}{2}}} + k^{\frac{\beta}{2}}$,
$\mathcal{M}_i=M- \frac{\mathcal{C} ^\beta_{A|B_{i+1}\cdots B_{N-1}}}{\mathcal{C}^\beta_{AB_i}}$ for $i=1,2, \cdots, m$, and
$\mathcal{Q}_j=M-\frac{\mathcal{C}^\beta_{AB_j}}{\mathcal{C} ^\beta_{A|B_{j+1}\cdots B_{N-1}}} $ for $j=m+1, \cdots , N-2$.
\end{theorem}

\begin{proof}
From the inequality \eqref{Con9}, we have
\begin{equation}\label{Con10}
\begin{array}{rl}
&\mathcal{C}_{A|B_1B_2\cdots B_{N-1}}^\beta\\[2.0mm]
&\ \ \ \geq \mathcal{C}^\beta_{AB_1}+ \Big[ M- \Big(\frac{\mathcal{C}_{A|B_2\cdots B_{N-1}}}{\mathcal{C}_{AB_1}} \Big)^\beta \Big]\mathcal{C}^\beta_{A|B_2\cdots B_{N-1}}\\[2.0mm]
&\ \ \geq\mathcal{C}_{AB_1}^\beta+\mathcal{M}_1 \mathcal{C}_{AB_2}^\beta + \mathcal{M}_1 \mathcal{M}_2\mathcal{C}_{A|B_3\cdots B_{N-1}}^\beta\\[2.0mm]
&\ \ \geq\cdots\\[2.0mm]
&\ \ \geq\mathcal{C}_{AB_1}^\beta+\mathcal{M}_1 \mathcal{C}_{AB_2}^\beta + \cdots +\mathcal{M}_1 \mathcal{M}_2 \cdots \mathcal{M}_{m-1}\mathcal{C}_{AB_m}^\beta  \\[2.0mm]
&\ \ \ \ +\mathcal{M}_1 \mathcal{M}_2 \cdots \mathcal{M}_{m} \mathcal{C}_{A|B_{m+1} \cdots B_{N-1}}^\beta.
\end{array}
\end{equation}
Similarly, as $\mathcal{C}_{AB_j}^2\leq k\mathcal{C}_{A|B_{j+1}\cdots B_{N-1}}^2$, for $j=m+1,\ldots,N-2$ we get
\begin{equation}\label{Con11}
\begin{array}{rl}
&\mathcal{C}_{A|B_{m+1}\cdots B_{N-1}}^\beta\\[2.0mm]
&\ \ \geq \mathcal{Q}_{m+1} \mathcal{C}_{AB_{m+1}}^\beta + \mathcal{C}_{A|B_{m+2\cdots B_{N-1}}}^\beta \\[2.0mm]
&\ \ \geq\cdots\\[2.0mm]
&\ \ \geq \mathcal{Q}_{m+1} \mathcal{C}_{AB_{m+1}}^\beta + \cdots +\mathcal{Q}_{N-2} \mathcal{C}_{AB_{N-2}}^\beta +\mathcal{C}_{AB_{N-1}}^\beta.
\end{array}
\end{equation}
Combining \eqref{Con10} and \eqref{Con11} we get the inequality \eqref{Con12}.
\end{proof}

In particular, if the conditions $k\mathcal{C}_{AB_i}\geq\mathcal{C}_{A|B_{i+1}\cdots B_{N-1}}$ for $i=1,2,\ldots,m$,
and $\mathcal{C}_{AB_j}^2\leq k\mathcal{C}_{A|B_{j+1}\cdots B_{N-1}}^2$ for $j=m+1,\ldots,N-2$, $\forall 1\leq m\leq N-3$, $N\geq 4$,
in Theorem \ref{concurrence1} reduce simply to $k\mathcal{C}_{AB_i}^2\geq \mathcal{C}_{A|B_{i+1}\cdots B_{N-1}}^2$ for $i=1,2,\ldots,N-2$,
then we have the following theorem.

\begin{theorem}\label{concurrence2}
For real number $0<k\leq 1$ and $N$-qubit mixed state $\rho_{AB_1\cdots B_{N-1}}$,
if $k\mathcal{C}_{AB_i}^2\geq\mathcal{C}_{A|B_{i+1}\cdots B_{N-1}}^2$ for all $i=1,2,\ldots,N-2$, then
\begin{equation}\label{Con13}
\mathcal{C}^\beta_{A|B_1\cdots B_{N-1}} \geq \sum\limits_{i=1}^{N-2} \Big( \prod\limits_{j=1}^i \mathcal{M}_j \Big) \mathcal{C}^{\beta}_{AB_{i+1}}
\end{equation}
for $\beta\geq 4$.
\end{theorem}

It can be seen that the inequalities \eqref{Con12} and \eqref{Con13} are tighter than the ones given in Refs. \cite{Yanglongmei} and \cite{Taoyuanhong},
since
$$
\frac{(1+k)^x-1}{k^x} + k^x-t^x \geq 2^x-t^x
$$
for $x\geq2$ and $0<k\leq1$. The equality holds when $k=1$. Namely, the result (\ref{Con4}) given in \cite{Taoyuanhong} is  a special case of our result for $k=1$.

\smallskip
\noindent{\bf Example 1} \, \ Consider the three-qubit state $|\psi\rangle_{ABC}$ in generalized Schmidt decomposition form \cite{Schmidt,SM.Fei5},
\begin{equation}\label{Con6}
\begin{array}{rcl}
|\psi\rangle_{ABC}&=&\lambda_0|000\rangle+\lambda_1e^{i\varphi}|100\rangle
+\lambda_2|101\rangle\\[1mm]
&&+\lambda_3|110\rangle+\lambda_4|111\rangle,
\end{array}
\end{equation}
where $\lambda_i\geq0$, $i=,1,2...,4$, and $\sum\limits_{i=0}^4\lambda_i^2=1$.
We have $\mathcal{C}_{AB}=2\lambda_0\lambda_2$, $\mathcal{C}_{AC}=2\lambda_0\lambda_3$ and $\mathcal{C}_{A|BC}=2\lambda_0\sqrt{\lambda_2^2+\lambda_3^2+\lambda_4^2}$
.
Set $\lambda_0=\lambda_3=\frac{1}{2},\ \lambda_2=\frac{\sqrt{2}}{2}$ and $\lambda_1=\lambda_4=0$.
We have $\mathcal{C}_{AB}=\frac{1}{2}$ , $\mathcal{C}_{AC}=\frac{1}{2\sqrt{2}}$ and $\mathcal{C}_{A|BC}=\frac{\sqrt{2}}{2}$ .
Hence we have the lower bounds of the inequalities \eqref{ylm}, \eqref{Con3}, \eqref{Con4} and \eqref{Con9}, respectively,
$$
\mathcal{C}_{AB}^\beta+\frac{(1+k)^\frac{\beta}{2}-1}{k^{\frac{\beta}{2}}}\mathcal{C}_{AC}^\beta=\big(\frac{1}{2}\big)^\beta+ \frac{(1+k)^\frac{\beta}{2}-1}{k^{\frac{\beta}{2}}}\big(\frac{1}{2\sqrt{2}}\big)^\beta,
$$
$$
\mathcal{C}_{AB}^\beta+\big(2^{\frac{\beta}{2}}-1\big)\mathcal{C}_{AC}^\beta=\big(\frac{1}{2}\big)^\beta+\big(2^{\frac{\beta}{2}}-1\big)\big(\frac{1}{2\sqrt{2}}\big)^\beta,
$$
$$\mathcal{C}_{AB}^\beta+\big[2^{\frac{\beta}{2}}-\big(\frac{\mathcal{C}_{AC}}{\mathcal{C}_{AB}}\big)^\beta
\big]\mathcal{C}_{AC}^\beta=\big(\frac{1}{2}\big)^\beta+ \big(2^{\frac{\beta}{2}}-\big(\frac{1}{\sqrt{2}}\big)^\beta\big)\big(\frac{1}{2\sqrt{2}}\big)^\beta
$$
and
$$
\mathcal{C}_{AB}^\beta+\big[\frac{(1+k)^\frac{\beta}{2}-1}{k^{\frac{\beta}{2}}} + k^{\frac{\beta}{2}}-\big(\frac{\mathcal{C}_{AC}}{\mathcal{C}_{AB}}\big)^\beta \big]\mathcal{C}_{AC}^\beta
$$
$$
=\big(\frac{1}{2}\big)^\beta+ \big[\frac{(1+k)^\frac{\beta}{2}-1}{k^{\frac{\beta}{2}}}+k^{\frac{\beta}{2}}-\big(\frac{1}{\sqrt{2}}\big)^\beta \big]\big(\frac{1}{2\sqrt{2}}\big)^\beta. $$
Set $k=0.8$. Fig. 1 shows that our result is better than the result in \cite{SM.Fei3,Taoyuanhong,Yanglongmei}. That is, the lower bound on $\mathcal{C}_{A|BC}^\beta$ that we derive from (\ref{Con9}) is larger than that derived from \eqref{Con4}-\eqref{Con4}.
\begin{figure}
\centering  
\includegraphics[width=8cm]{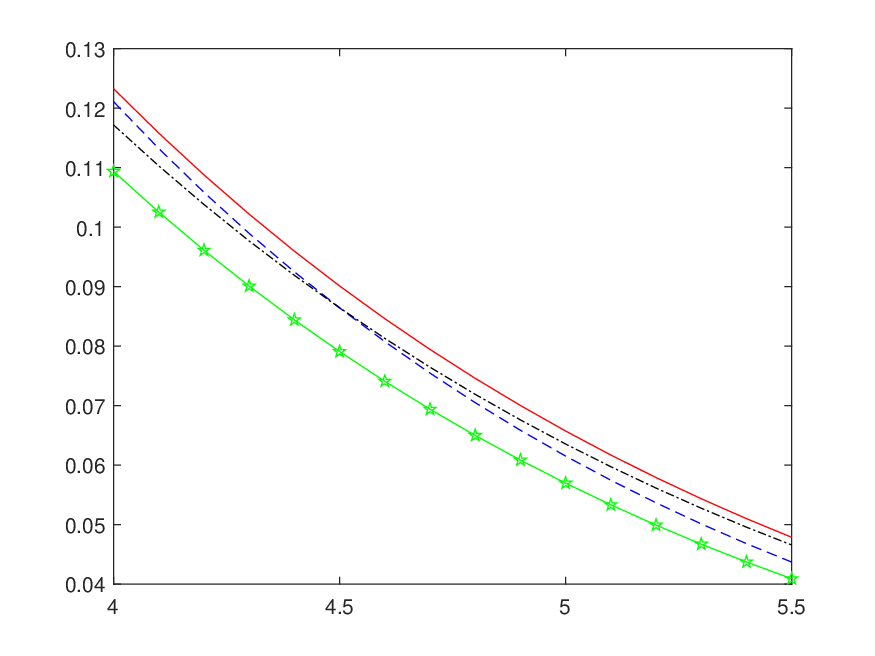}
\caption{The vertical axis is the lower bound of the concurrence of  $|\psi\rangle_{ABC}$. The horizontal axis is $\beta$.
The solid-red line represents the lower bound from our result. The dashed-blue line represents the lower bound of (\ref{Con4}) from \cite{Taoyuanhong}. The star-solid-green line represents the lower bound of (\ref{Con3}) from \cite{SM.Fei3}. The dot-dashed-black line represents the lower bound of  (\ref{ylm}) from \cite{Yanglongmei}}.
\end{figure}

\section{Tighter monogamy inequalities related to negativity}
 The negativity is another well-known quantifier of bipartite entanglement. The negativity of a state $\rho_{AB}$ is defined as
$\mathcal{N}(\rho_{AB})=\big(\|\rho_{AB}^{T_A}\|-1\big)/2$ \cite{Vidal}, where $\rho_{AB}^{T_A}$ is the partial transposed matrix of $\rho_{AB}$ with respect to the subsystem $A$,
and $\|X\|$ is the trace norm of $X$, i.e., $\|x\|={\rm tr}\sqrt{xx^{\dag}}$.
For simplicity, we use the definition of negativity as $\|\rho_{AB}^{T_A}\|-1$.
Particularly, for any bipartite pure state $|\psi\rangle_{AB}$,
$\mathcal{N}(|\psi\rangle_{AB})=2\sum\limits_{i<j}\sqrt{\lambda_i\lambda_j}=({\rm tr}\sqrt{\rho_A})^2-1$,
where $\lambda_i s$ are the eigenvalues of the reduced density matrix $\rho_A={\rm tr}_B|\psi\rangle_{AB}\langle\psi|$ \cite{HTAG}.
The convex-roof extended negativity (CREN) of a mixed state $\rho_{AB}$ is defined by
\begin{equation}
\mathcal{N}_c(\rho_{AB})=\min\limits_{\{p_i,|\psi_i\rangle\}}\sum_{i}p_i\mathcal{N}(|\psi_i\rangle),
\end{equation}
where the minimum is taken over all possible pure state decompositions of $\rho_{AB}$. For any two-qubit mixed state $\rho_{AB}$ one has
$\mathcal{N}_c(\rho_{AB})=\mathcal{C}(\rho_{AB})$ \cite{JSK}.
Similar to the concurrence, we have the following Theorems.

\begin{theorem}
For any real number $0<k\leq 1$ and $N$-qubit mixed state $\rho_{AB_1\cdots B_{N-1}}$,
if $k\mathcal{N}_{cAB_i}^2\geq\mathcal{N}_{cA|B_{i+1}\cdots B_{N-1}}^2$
for $i=1,2,\ldots,m$, and $\mathcal{N}_{cAB_j}^2\leq k\mathcal{N}_{cA|B_{j+1}\cdots B_{N-1}}^2$ for $j=m+1,\ldots,N-2$,
$\forall 1\leq m\leq N-3$, $N\geq 4$,
\begin{equation}\label{negativity1}
\begin{aligned}
\mathcal{N} ^{\beta}_{cA|B_1\cdots B_{N-1}} & \geq \sum\limits_{i=1}^{m-1} \Big( \prod\limits_{j=1}^i \mathcal{P}_j \Big) \mathcal{N}^{\beta}_{cAB_{i+1}}   \\
& +\Big( \prod\limits_{i=1}^{m} \mathcal{P}_i \Big) \Big( \prod\limits_{j=m+1}^{N-2} \mathcal{R}_j \mathcal{N}^\beta_{cAB_j} + \mathcal{N}^{\beta}_{cAB_{N-1}} \Big)
\end{aligned}
\end{equation}
for all $\beta\geq 4$, where $M=\frac{(1+k)^{\frac{\beta}{2}}-1}{k^{\frac{\beta}{2}}} + k^{\frac{\beta}{2}}$,
$\mathcal{P}_i=M- \frac{\mathcal{N} ^\beta_{cA|B_{i+1}\cdots B_{N-1}}}{\mathcal{N}^\beta_{cAB_i}}$ for $i=1,2, \cdots, m$, and
$\mathcal{R}_j=M-\frac{\mathcal{N}^\beta_{cAB_j}}{\mathcal{N} ^\beta_{cA|B_{j+1}\cdots B_{N-1}}} $ for $j=m+1, \cdots , N-2$.
\end{theorem}

\begin{theorem}
For any real number $0<k\leq 1$ and $N$-qubit mixed state $\rho_{AB_1\cdots B_{N-1}}$,
if $k\mathcal{N}_{cAB_i}^2\geq\mathcal{N}_{cA|B_{i+1}\cdots B_{N-1}}^2$ for all $i=1,2,\ldots,N-2$, then
\begin{equation}\label{negativity2}
\mathcal{N}^\beta_{cA|B_1\cdots B_{N-1}} \geq \sum\limits_{i=1}^{N-2} \Big( \prod\limits_{j=1}^i \mathcal{P}_j \Big) \mathcal{N}^{\beta}_{cAB_{i+1}}
\end{equation}
for $\beta\geq 4$, where $M$ and $\mathcal{P}_j$ are defined in Theorem 3.
\end{theorem}

\smallskip
\noindent{\bf Example 2} \, \
Consider the three-qubit state (\ref{Con6}) in Example 1 again. Set $\lambda_0=\lambda_1=\lambda_2=\frac{\sqrt{2}}{3}$ and $\lambda_3=\lambda_4=\frac{1}{\sqrt{6}}$.
We have $\mathcal{N}_{cAB}=\frac{4}{9}$ , $\mathcal{N}_{cAC}=\frac{2\sqrt{3}}{9}$ and $\mathcal{N}_{cA|BC}=\frac{2\sqrt{10}}{9}$ . Hence, we have the lower bounds of the inequalities \eqref{ylm}, \eqref{Con3}, \eqref{Con4} and \eqref{Con9}, respectively,
$$\mathcal{N}_{cAB}^\beta+\frac{(1+k)^\frac{\beta}{2}-1}
{k^{\frac{\beta}{2}}}\mathcal{N}_{cAC}^\beta=\big(\frac{4}{9}\big)^\beta+ \frac{(1+k)^\frac{\beta}{2}-1}{k^{\frac{\beta}{2}}}\big(\frac{2\sqrt{3}}{9}\big)^\beta,
$$
$$
\mathcal{N}_{cAB}^\beta+\big(2^{\frac{\beta}{2}}-1\big)
\mathcal{N}_{cAC}^\beta=\big(\frac{4}{9}\big)^\beta
+\big(2^{\frac{\beta}{2}}-1\big)\big(\frac{2\sqrt{3}}{9}\big)^\beta,
$$
$$
\mathcal{N}_{cAB}^\beta+\big[2^{\frac{\beta}{2}}
-\big(\frac{\mathcal{N}_{cAC}}{\mathcal{N}_{cAB}}\big)^\beta \big]\mathcal{N}_{cAC}^\beta=\big(\frac{4}{9}\big)^\beta+ \big(2^{\frac{\beta}{2}}-\big(\frac{\sqrt{3}}{2}\big)^\beta\big)
\big(\frac{2\sqrt{3}}{9}\big)^\beta
$$
and
$$
\mathcal{N}_{cAB}^\beta+\big[\frac{(1+k)^\frac{\beta}{2}-1}{k^{\frac{\beta}{2}}} + k^{\frac{\beta}{2}}-\big(\frac{\mathcal{N}_{cAC}}{\mathcal{N}_{cAB}}\big)^\beta \big]\mathcal{N}_{cAC}^\beta
$$
$$
=\big(\frac{4}{9}\big)^\beta+ \big[\frac{(1+k)^\frac{\beta}{2}-1}{k^{\frac{\beta}{2}}}+k^{\frac{\beta}{2}}-\big(\frac{\sqrt{3}}{2}\big)^\beta \big]\big(\frac{2\sqrt{3}}{9}\big)^\beta.
$$
Our result is better than the results in \cite{SM.Fei3, Yanglongmei, Taoyuanhong} for $\beta\geq4$, see Fig. 2 for $k=0.8$.
That is, the lower bound on $\mathcal{N}_{cA|BC}^\beta$ ($\mathcal{C}_{A|BC}^\beta$)that we derive from (\ref{negativity1}) is larger than that derived from \cite{SM.Fei3, Yanglongmei, Taoyuanhong}.
\begin{figure}
\centering
\includegraphics[width=8cm]{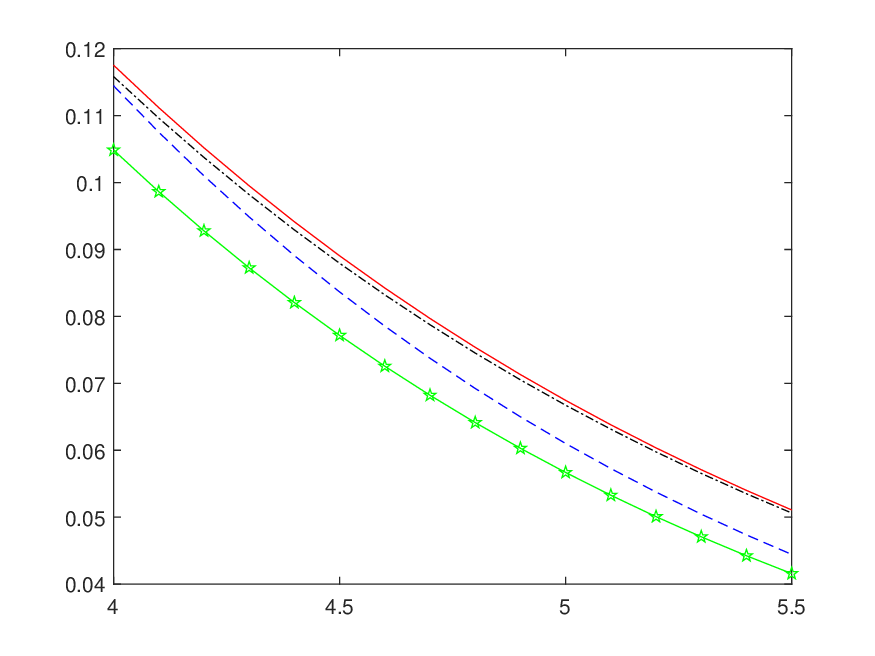}
\caption{The vertical axis is the lower bound of the concurrence of $|\psi\rangle_{ABC}$. The horizontal axis is $\beta$.
The solid-red line represents the lower bound from our result. The dashed-blue line represents the lower bound of (\ref{Con4}) from \cite{Taoyuanhong}. The star-solid-green line represents the lower bound of (\ref{Con3}) from \cite{SM.Fei3}. The dot-dashed-black line represents the lower bound of  (\ref{ylm}) from \cite{Yanglongmei}}.
\end{figure}

\section{Conclusion}
 Entanglement monogamy is a crucial aspect of multipartite entanglement. We have presented monogamy relations based on the $\beta$th power of concurrence and negativity. It should be noted that all of the monogamy inequalities proposed in this paper are tighter than the previous related ones. These improved monogamy inequalities offer more precise descriptions of how entanglement is distributed within multiqubit systems. Our approaches can be applied to the study on the monogamy properties based on other quantum correlation measures as well as quantum coherence \cite{framework,coherence}.

\section{Acknowledgements*} This work is supported by NSFC (Grant Nos. 12075159, 12171044), Beijing Natural Science Foundation (Z190005), and the Academician Innovation Platform of Hainan Province.

\bigskip
\noindent{\bf Data availability statement}
All data generated or analysed during this study are included in this published article.

\end{document}